\documentclass[conference]{IEEEtran}

\usepackage[english]{babel}
\usepackage[utf8x]{inputenc}
\usepackage{fancyvrb}
\usepackage[T1]{fontenc}
\usepackage{booktabs} 
\usepackage{float}
\usepackage{graphicx}
\usepackage{amsmath}
\usepackage{amsfonts}
\usepackage{amssymb}
\usepackage{listings}
\usepackage{array}
\usepackage{setspace}
\usepackage{latexsym,subeqnarray,amsmath}
\usepackage[a4paper,top=3cm,bottom=2cm,left=3cm,right=3cm,marginparwidth=1.75cm]{geometry}

\usepackage{amsmath,graphicx}
\usepackage{amsfonts,amsthm}
\usepackage{amssymb}
\usepackage{graphicx}
\usepackage[colorinlistoftodos]{todonotes}
\usepackage[colorlinks=true, allcolors=blue]{hyperref}

\newtheorem{proposition}{Proposition}

\title{OpenMP Parallelization of Dynamic Programming and Greedy Algorithms}
\author{Claude Tadonki\\Mines ParisTech - PSL Research University \\
Centre de Recherche en Informatique (CRI)\\
35, rue Saint-Honor\'e, 77305, Fontainebleau Cedex (France)\\
Email: claude.tadonki@mines-paristech.fr}

\begin{document}
\maketitle
\begin{abstract}
Multicore has emerged as a typical architecture model since its advent and stands now as a standard. The trend is to increase the number of cores and improve the performance of the memory system. Providing an efficient multicore implementation for a important algorithmic kernel is a genuine contribution. From a methodology standpoint, this should be done at the level of the underlying paradigm if any. In this paper, we study the cases of {\em dynamic programming} and {\em greedy algorithms}, which are two major algorithmic paradigms. We exclusively consider directives-based loop parallelization through OpenMP and investigate necessary pre-transformations to reach a regular parallel form. 
We evaluate our methodology with a selection of well-known combinatorial optimization problems on an INTEL Broadwell processor. Key points for scalability are discussed before and after experimental results. Our immediate perspective is to extend our study to the manycore case, with a special focus on NUMA configurations.  
\end{abstract}
 
\section{Introduction}
With the advent and pervasiveness of multicore processors, designing shared memory parallel programs is on the way to routine consideration. OpenMP\cite{openmp} currently stands as a standard for multicore parallelization and genuine efforts are made to make it as powerful as expected. However, the case of irregular algorithms is problematic because of load imbalance at runtime. In addition, programs that have a sequential profile in their original form need appropriate code transformation in order to expose parallelism. This essential preprocessing might out of the skills of an ordinary programmer or might be subject to some reluctance from experts. The main concern is to keep the advantage of a short {\em time-to-code} behind the use of OpenMP, while trying to get a reasonably efficient parallel implementation. This is the context of our work, where we consider two major algorithmic paradigms\cite{gr-dp} namely {\em dynamic programming} and {\em greedy algorithms}.  

{\em Dynamic programming} and {\em greedy algorithms} are widely used to design efficient algorithms for combinatorial optimization problems. The corresponding algorithms have a structural regularity that does not always correspond to the expected regularity neither for the iteration spaces nor for memory accesses. In addition, appropriate loop transformations\cite{loop} should be applied before the parallelization. We investigate this parallel design concern through a selection of well-known problems: {\em shortest paths}\cite{warshall}, {\em graph flooding}\cite{flood}, {\em 0-1 knapsack}\cite{knapsack, knapsack2}, {\em longest common subsequence}\cite{lcs}, {\em longest increasing subsequence}\cite{lis}, and {\em minimum spanning tree}\cite{mst}. For each case study, we explain and discuss the necessary transformation then provide the corresponding pseudo-code. We do not intend to provide a state-of-the-art solution for each problem, but we rather focus on a generic methodology, which is the goal behind this work. The reader should them see it as a step towards a generic parallelization methodology for the two considered paradigms.

The rest of the paper is organized as follows. The next two section explore successively the case of dynamic programming and then greedy algorithms. Each of both sections shortly describes the paradigm, then the selected problems and their OpenMP parallelization, and finished with experimental results. Section 4 concludes the paper.

\section{The case of dynamic programming}
\subsection{Definition and selected cases}
Dynamic programming (usually referred to as DP ) is a well-known paradigm mostly used for (but not restricted to) discrete optimization problems. From a given input $S$, dynamic programming works iteratively in a finite number of computing steps of the form
\begin{equation}
S_{k+1}=f(k,S_k),
\end{equation}
where $f$ is the generic iteration function and $k$ the iterator parameter. It is common to consider in-place computation, thus the procedure works by means of iterative updates. Table \ref{tab_dp} provides a selection of well-known dynamic programming cases.
\begin{table*}
\begin{tabular}{|r|l|l|l|l|l|}
\hline
{\bf N°} & {\bf Problem} & {\bf Algorithm} & {\bf Generic Update}\\
\hline
1 & Shortest Paths & Floyd-Warshall & $m_{i,j}=\min (m_{i,j},m_{i,k}+m_{k,j})$\\
\hline
2 & Dominated Graph Flooding & Berge & $\tau_i = \min(\tau_i, max(v_{i,j}, \tau_j)$\\
\hline
3 & 0-1 Knapsack Problem & Standard DP & $t_{i,w} = \max(t_{i-1,w},v_{i-1}+t_{i-1,w-w_i})$\\
\hline
4 & Longest Common Subsequence & Standard DP & $c_{i,j} =  
\left\{
\begin{array}{ll}
c_{i-1,j-1} + 1 & \mbox{if } (s_i=t_i),    \\
\max(c_{i-1,j},c_{i,j-1}) & \mbox{otherwise}\\
\end{array}
\right.$\\
\hline
5 & Longest Increasing Subsequence & Standard DP & $l_i = \max(l_i, l_j+1) \mbox{ if } (a_i>a_j)$\\
\hline
\end{tabular}
\caption{\label{tab_dp}Selected dynamic programming cases}
\end{table*}

We now briefly describe each of the selected problems and provide an ordinary shared memory parallelization. 
\subsection{The 0-1 Knapsack Problem}
Given a set of items, each with a mass and a value, the Knapsack Problem is to determine the number of each item (0 or 1 for the 0-1 Knapsack) to include in a collection so that the total weight is less than or equal to a given limit and the total value is as large as possible. It derives its name from the problem faced by someone who is constrained by a fixed-size knapsack and must fill it with the most valuable items. An ordinary dynamic programming procedure to get the solution is given by the algorithm in figure \ref{knp}, where $n$ is the total number of items, $v$ (resp. $w$) the array of corresponding values (resp. weight), and $W$ the weight limit of the selected items. The solution is constructed inside matrix $V$, whose first row (resp. first column) is initialized with $0$. The value in $V$[$i,j$] represents the solution of the subproblem with items $\{1,2,\cdots, i\}$ and total weight $j$. The final solution is thus in $V$[$n,W$], which is the total value of the selected subset.
\begin{figure}[H]
{\bf Knapsack}($v,w,n,W$)\{ 

\ \ {\bf for}($i=1;i\leq n;j$++) 

\ \ \ \ {\bf for}($j=1;j\leq W;j$++) 

\ \ \ \ \ \ {\bf if}($w$[$i$]$\leq j$)

\ \ \ \ \ \ \ \ $V$[$i,j$]$=\max$\{ $V$[$i-1,j$],$v$[$i$]+$V$[$i-1,j-w$[$i$]]  \};

\ \ \ \ \ \ {\bf else}

\ \ \ \ \ \ \ \ $V$[$i,j$]=$V$[$i-1,j$];

\ \ {\bf return} $V$[$n,W$];

\}
\caption{\label{knp}Dynamic programming algorithm for the 0-1 Knapsack Problem}
\end{figure}

We see that all dependencies are of the form $(i,j)\leftarrow (i-1,j-\lambda)$, which guarantees that for a fix $i$, all updates along $j$-axe can be performed in parallel (provided all calculations at the level of $i-1$ has been completed). In addition, the one-step lifetime of variables $V$($i,:$) suggest to compress along $i$-axe by storing $V$($i,:$) at $V$($i\bmod 2,:$), thus using a $2W$ array instead of $n\times W$. This yields the OpenMP version displayed in figure \ref{knp-omp}. 
\begin{figure}[H]
{\bf Knapsack}($v,w,n,W$)\{ 

\ \ {\bf for}($i=1;i\leq n;j$++) 
        
\ \ \ \ {\tt \#pragma omp parallel for}
        
\ \ \ \ {\bf for}($j=1;j\leq W;j$++) 

\ \ \ \ \ \ {\bf if}($w$[$i$]$\leq j$)

\ \ \ \ \ \ \ \ $V$[$i\%2,j$]$=\max$\{ $V$[$(i-1)\%2,j$],$v$[$i$]+$V$[$i-1,j-w$[$i$]]  \};

\ \ \ \ \ \ {\bf else}

\ \ \ \ \ \ \ \ $V$[$\%2i,j$]=$V$[$(i-1)\%2,j$];

\ \ {\bf return} $V$[$n,W$];

\}
\caption{\label{knp-omp}OpenMP Loop for the 0-1 Knapsack Problem}
\end{figure}
\subsection{Dominated Graph Flooding}
Given weighted undirected graph $G=(X,E,v)$ and a ceiling function $\omega:X\rightarrow \mathbb{R}$. A valid flooding function of $G$ under the ceiling constraint $\omega$ is the \underline{maximal} function $\tau:X\rightarrow \mathbb{R}$ satisfying
\begin{equation}\label{flood-canon}
\forall x,y\in X:\quad \tau(x) \leq \min (\max(v(x,y), \tau(y)) \mbox{, } \omega(x)).
\end{equation}
It can be shown that such a function $\tau$ satisfies equation (\ref{flood-opt})
\begin{equation}\label{flood-opt}
\forall x,y\in X:\quad \tau(x) = \min (\max(v(x,y), \tau(y)) \mbox{, } \omega(x)).
\end{equation}
A dynamic programming algorithm to compute $\tau$ was proposed by C. Berge. For a given graph with $n$ vertices, valuation $v=(v_{ij})$ and ceiling $\omega=(\omega_{ij})$, Berge algorithms computes the flooding $\tau=(\tau_{i})$ as follows:\\
(i) $\tau^{(0)}\leftarrow \omega$ \\
(ii) repeat update (\ref{flood-berge}) until ($\tau_i^{(k)}=\tau_i^{(k-1)}$)
\begin{equation}\label{flood-berge}
\tau_i^{(k)} = \min(\tau_i^{(k)}, max(v_{i,j}, \tau_j^{(k-1)}), i=1,2,\cdots, n.
\end{equation}
The corresponding code, where $\tau_i^{(k)}$ is stored at $\tau(k\bmod 2,:)$, is provided in figure \ref{berge-omp}. The computations of the components of $\tau_i^{(k)}$ are independent of each other, thus the corresponding loop can be freely parallelized as done through OpenMP directive.   
\begin{figure}[H]
\begin{Verbatim}[fontsize=\small]
while(doIt==1){
  #pragma omp parallel for private(j)
  for(i=0;i<n;i++){
    h[k%2,i] = h[(k+1)%2,i]; 
    for(j=0;j<n;j++)
      h[k%2,i] = min(h[k%2,i], 
               max(G[i,j],h[(k+1)%2,j]));
  }
  doIt=0; 
  for(i=0;i<n;i++) 
     if(h[1,i] != h[0,i]) {doIt=1; break;}
  k++;
}
\end{Verbatim}
\caption{\label{berge-omp}OpenMP Loop for the Graph Flooding Problem}
\end{figure}

\subsection{Shortest Paths}
This a well-know classical graph problem. The problem is to find shortest distances between every pair of vertices in a given edge-weighted directed Graph, which does not contain any cycles of negative length. For a given graph of order $n$, represented by a $n\times n$ distances matrix $M=(m_{ij})$, where $m_{ii}=0$ and $m_{ij}=+\infty$ if there is no connection between $i$ and $j$, the Floyd-Warshall algorithm iteratively computes the matrices $M^{(k)}$ of the shortest paths that only consider the vertices in $\{1,2,\cdots, k\}$, $k=1,2,\cdots, n$. The corresponding code considering in-place computation is provided in figure \ref{fw}, where the computation of each step $k$ is executed in parallel (row-wise).
\begin{figure}[H]
\begin{Verbatim}[fontsize=\small]
for(k=0;k<n;k++)
 #pragma omp parallel for private(j)
 for(i=0;i<n;i++)
  for(j=0;j<n;j++)
   M[i,j] = min(M[i,j],M[i,k]+M[k,j]);
\end{Verbatim}
\caption{\label{fw}OpenMP Loop for the Floyd-Warshall procedure}
\end{figure}

Since row $k$ and column $k$ (the pivots) remain unchanged after step $k$, the corresponding loop can be executed in parallel. 

\subsection{Longest Common Subsequence}
Given two finite sequences of numbers, the {\em Longest Common Subsequence} (LCS) problem is to find the (length of the) longest common contiguous subsequence. The problem is commonly related to strings. A basic dynamic programming algorithm for this problem proceeds as follows. Given two sequences $(u_i)_{i=1,\cdots,n}$ and $(v_i)_{i=1,\cdots,m}$, we define $c_{ij}$ as the length of the LCS in $(u_1,\cdots,u_i)$ and $(v_1,\cdots,v_j)$. We have
\begin{equation}
c_{ij} = 
\left\{
\begin{array}{ll}
    c_{i-1, j-1}+1 &if \mbox{ } u_i=v_j\\
    \max(c_{i-1, j},c_{i, j-1})&otherwise
\end{array}
\right.
\end{equation}
This yields the code of figure \ref{lcs}.
\begin{figure}[H]
\begin{Verbatim}[fontsize=\small]
for(i=1;i<n;i++)
 for(j=1;j<n;j++)
  if(S[i] == T[j]) c[i,j] = c[i-1,j-1]+1;
  else c[i,j] = max(c[i,j-1],c[i-1,j]);
\end{Verbatim}
\caption{\label{lcs}Loop for the LCS procedure}
\end{figure}
As it is, the loop cannot be parallelized. Indeed, the dependence $(i,j)\leftarrow (i-1,j-1)$ constraints both $i$ and $j$ axes. One way to overcome this is to consider the loop skewing transformation, were the computation are done following the hyperplanes $i+j=k$, $k=2,\cdots, 2(n-1)$, each of which being parallel. This yields the OpenMP code of figure \ref{lcs-omp}. 
\begin{figure}[H]
\begin{Verbatim}[fontsize=\small]
for(k=2;k<=n;k++)
  #pragma omp parallel for
  for(i=1;i<k;i++){
   if(S[i] == T[(k-i)]) 
     c[w(i,k-i)] = c[w(i-1,(k-i)-1)]+1;
   else 
     c[w(i,(k-i))] = max(c[w(i,(k-i)-1)],
                        c[w(i-1,(k-i))]);
  }
for(k=n+1;k<=2*(n-1);k++)
  #pragma omp parallel for
  for(i=(k-n)+1;i<n;i++){
   if(S[i] == T[k-i]) 
      c[i,k-i] = c[i-1,(k-i)-1]+1;
   else 
      c[i,(k-i)] = max(c[i,(k-i)-1],
                          c[i-1,(k-i)]);
  }
\end{Verbatim}
\caption{\label{lcs-omp}OpenMP Loop for the LCS procedure}
\end{figure}

\subsection{Longest Increasing Subsequence}
Given a finite sequence of numbers, the {\em Longest Increasing Subsequence} (LIS) problem is to find the (length of the) longest of its subsequences. The basic idea of a dynamic programming for this case is that, given an increasing subsequence and a new element out of it, we can form a new increasing subsequence (with one more element) if that element is greater that the last element (the greatest) of the subsequence. Thus, with a sequence of $n$ numbers $a_1,a_2,\cdots,a_n$, if we define $l_i$ as length of the LIS restricted to $a_1,a_2,\cdots,a_i$ and ending with $a_i$, then we have
\begin{equation}
l_i = \max_{\{1\leq j\leq i-1:a_i>a_j\}}(l_j+1).
\end{equation}
This yields the dynamic programming procedure described in figure \ref{lis}, where {\tt LS[]} is initialized with {\tt 1}. Note that the global solution is the maximum of {\tt LS[]}.
 \begin{figure}[H]
\begin{Verbatim}[fontsize=\small]
for(i=0;i<n;i++)
  for(j=0;j<i;j++)
    if(a[i] > a[j]) 
       LS[i] = max(LS[i], LS[j]+1);
\end{Verbatim}
\caption{\label{lis}Dynamic programming loop for the LIS problem}
\end{figure}
We can see that the process is strongly sequential like the {\em prefix  computation}\cite{-} and none of the loop levels in figure \ref{lis} can be parallelized. To fix this, we consider:\\
$l_i$: length of the longest increasing subsequence going to (ending with) $a_i$\\
$s_i$: length of the longest increasing subsequence coming from (starting with) $a_i$\\
$d_i$: length of the longest increasing subsequence passing through (containing) $a_i$\\
We have
\begin{equation}
s_i = \max_{\{i-1\leq j\leq n:a_i<a_j\}}(l_j+1),
\end{equation}
and
\begin{equation}
d_i = \max_{\{1\leq j\leq i-1:a_i>a_j\}}(l_j+s_i).
\end{equation}
For a given $k$, $1<k<n$, for $i\in \{k+1,\cdots, n\}$, we define:\\ $d^{(k)}_i$: length of the LIS passing through $a_i$ excluding items in $\{ a_{k+1},\cdots,a_{i-1}\}$. We have
\begin{equation}
d^{(k)}_i = \max_{\{1\leq j\leq k:a_i>a_j\}}(l_j+s_i).
\end{equation}

\begin{proposition} For any $k$, $1<k<n$, we have 
$$\max_{\{k< i\leq n\}} \{d_i\} = \max_{\{k< i\leq n\}} \{d^{(k)}_i\}$$.
\end{proposition}
\begin{proof}
It is obvious that $d_i\geq d^{(k)}_i$ since $d^{(k)}_i$ considers a subset of the values related to $d_i$. Thus 
\begin{equation}
\max_{\{k< i\leq n\}} \{d_i\} \geq \max_{\{k< i\leq n\}} \{d^{(k)}_i\}.
\end{equation}
For any $i$, $k<i\leq n$,  let show that there is $j$, $k<j\leq n$, such that $d^{(k)}_j\geq d_i$. Let $i\in \{k+1,\cdots,n\}$:
\begin{itemize}
\item If the longest increasing subsequence passing through $a_i$ does not contains any element in $\{ a_{k+1},\cdots,a_{i-1}\}$ then $d^{(k)}_i= d_i$ (i.e $j=i$).
\item Otherwise, let $j$ be the smallest index in  $\{ k,\cdots,{i-1}\}$ such that $a_j$ belongs to he longest increasing subsequence passing through $a_i$. Thus, this subsequence does not contains any elements in  $\{ a_{k+1},\cdots,a_{j-1}\}$, we have $d^{(k)}_j= d_i$. 
\end{itemize}
Thus, for any $i$, there is $j$ such that $d^{(k)}_j\geq d_i$, which leads
\begin{equation}
\max_{\{k< i\leq n\}} \{d_i\} \leq \max_{\{k< i\leq n\}} \{d^{(k)}_i\}.
\end{equation}
\end{proof}
\noindent Note that $d^{(k)}_i$, $k$ fixed and $i=k+1,k+2,\cdots, n$, are independent to each other, thus can be computed in parallel.
Since the length of the LIS is given by
\begin{equation}
\max(\max_{\{1\leq i\leq k\}}\{l_i\},\max_{\{k< i\leq n\}}\{d^{(k)}_i\}),
\end{equation}
the steps of algorithm to compute the LIS is
\begin{itemize}
\item compute $l_i$, $i=1,2,\cdots, k$
\item compute $s_i$, $i=k+1,k+2,\cdots, n$
\item compute $d^{(k)}_i$, $i=k+1,k+2,\cdots, n$
\item compute $\max(\max_{\{1\leq i\leq k\}}\{l_i\},\max_{\{k< i\leq n\}}\{d^{(k)}_i\})$
\end{itemize}
The first two steps can be perform independently, and the last step is fully parallel. This yields the OpenMP code provided in figure \ref{lis-omp}.
\begin{figure}[H]
\begin{Verbatim}[fontsize=\small]
#pragma omp sections private(i,j)
{		
  #pragma omp section
  for(i=0;i<n/2;i++)
   for(j=0;j<i;j++)
     if(a[i] > a[j]) 
        LS[i] = max(LS[i], LS[j]+1);
            
  #pragma omp section
  for(i=n-1;i>=n/2;i--)
    for(j=n-1;j>i;j--)
      if(a[i] < a[j]) 
         LS[i] = max(LS[i], LS[j]+1);
} 
#pragma omp parallel for  private(v,j)
  for(i=n/2;i<n;i++){
     v = LS[i];
     for(j=0;j<n/2;j++)
     if(a[i] > a[j]) 
        LS[i] = max(LS[i], LS[j]+v);        
  }
\end{Verbatim}
\caption{\label{lis-omp}OpenMP parallelization the LIS problem}
\end{figure}
\subsection{Performance evaluation and related observations}
We evaluate our methodology for dynamic programming on our selected case studies using height 2.20 GHz-cores of an INTEL Broadwell E/P. The graphs are randomly generated with various sizes and different levels of density. Table \ref{res_dp} displays our experimental results.
\begin{table*}
\begin{tabular}{|l|r||r|r|r|r|r|r|r|r|r|}
\cline{4-10}
\multicolumn{3}{c}{} & \multicolumn{7}{|c|}{Number of cores (speedup)}\\   
\hline
{\bf Problem} & {\bf N} & {\bf Seq T(s)} & {\bf 2} & {\bf 3} & {\bf 4} & {\bf 5} & {\bf 6} & {\bf 7}& {\bf 8}\\
\hline
KNAPSACK & 10000 &   1.422  & 1.97 & 2.93 & 3.86 & 4.76 & 5.60 & 6.44 & 7.19 \\
  \hline
WARSHALL &  1000 &   0.942  & 1.99 & 2.98 & 3.96 & 4.94 & 5.90 & 6.86 & 7.81 \\
  \hline
LIS &10000 &   0.205  & 1.35 & 1.52 & 1.63 & 1.70 & 1.75 & 1.79 & 1.82 \\
  \hline     
LCS &10000 &   0.575  & 2.00 & 3.15 & 4.28 & 5.19 & 5.77 & 6.26 & 6.62 \\
  \hline
BERGE &1000 &   0.022  & 1.99 & 2.96 & 3.94 & 4.89 & 5.84 & 6.66 & 7.49 \\
\hline
\end{tabular}
\caption{\label{res_dp}Experimental results our dynamic programming parallelization}
\end{table*}
We can see that speedups are quite good for all cases except the LIS, which is bounded by its strongly sequential part despite our transformation for a better parallelization. The maximum speedup for the LIS following the analysis of our parallelization seems to be 2 and we can see that we are moving to that limit. We emphasize on the fact that we are in a context of directives-based parallelization, which is more simpler from the programming standpoint but has a natural limitation exacerbated with irregular or strongly sequential applications.   

\section{The case of greedy algorithm}
\subsection{Definition and selected cases}
Greedy algorithm is an algorithmic paradigm mainly used for discrete optimization problems. The basic idea is to iteratively populate the solution space by adding the best known candidate at each step. From the algorithmic viewpoint, the key is the selection process, which should definitely lead to the expected solution. From the complexity viewpoint, the key is the efficiency of the selection, which should be implemented at the best (memory and/or processing time). From a given input set $E$, the generic step of a greedy algorithm is of the form
\begin{equation}
S_{k+1}=S_k\cup f(E-S_k),
\end{equation}
where $f$ is the generic selection function.  Table \ref{tab_ga} provides a selection of well-known greedy algorithms.
\begin{table*}
\begin{tabular}{|r|l|l|l|l|l|}
\hline
{\bf N°} & {\bf Problem} & {\bf Algorithm} & {\bf Generic Selection}\\
\hline
1 & Shortest Paths (from a source node $s$) & Dijkstra & $\displaystyle i_{k+1} = \min_{i\in E-S_k}dist(s,i)$\\
\hline
2 & Minimum Spanning Tree & Prim & $\displaystyle a_{k+1} = \min_{i\in S_k,j\in E-S_k}m_{i,j}$\\
\hline
3 & Dominated Graph Flooding & Moore-Dijkstra & $\displaystyle i_{k+1} = \min_{i\in E-S_k}\tau_i$\\
\hline
\end{tabular}
\caption{\label{tab_ga}Selected greedy algorithm cases}
\end{table*}
Typically, each selection is followed by a generic update of the pivot information (Prim's algorithm does not have one). We now examine each of the selected cases.
\subsection{Shortest Paths}
The problem here is to compute the shortest paths from a given fixed node. The greedy algorithm from Dijkstra select the closest node from the remaining ones regardless of the valuations of their arcs. Thus, the distances (from the source node) are updated through the inspection of the potential changes induced by the selected node. Figure \ref{gsp}  illustrates the main phase of the algorithm. We assume that the source node is $0$ and distances vector $d$ has been initialized with $0$. In addition, as we want to have the range {\tt [1..p-1]} (resp. {\tt [p..n-1]}) for the set of selected (resp. remaining) nodes, we use an array {\tt nd} such that {\tt nd[k]} is the id of the node at position {\tt k}, thus the corresponding indirections.  The ultimate inner loop is the update of the distances after the new selection.
 \begin{figure}[H]
\begin{Verbatim}[fontsize=\small]
for(p=1;p<n;p++){
  k = p;
  for(i=p;i<n;i++)
    if(d[nd[i]] < d[nd[k]]) k = i;
  swap(nd[k],nd[p]);
  k = nd[p]; 
  for(j=0;j<deg[k];j++){
    i = g[k][j];
    if(d[i] > d[k]+m[i,k]) d[i] = d[k]+m[i,k];
  }
}
\end{Verbatim}
\caption{\label{gsp}Dijkstra greedy algorithm for the shortest paths}
\end{figure}
A close look at this algorithm allows to realize that only the update loop (the last inner loop) can be parallelized directly. This might be typical with greedy algorithm, where the update phase is usually individual. For a better efficiency, the selection loop can be made (directly) parallelizable through explicit blocking as follows. The whole search space is divided into equal size blocks. Then, the selection is made within each block and stored in a (global) array at a position that corresponds to the id of that block. Afterwards, a final reduction is made through the global array of local selections in order to get the global one. The corresponding transformation is expressed in figure \ref{psel}.
 \begin{figure}[H]
\begin{Verbatim}[fontsize=\small]
 for(j=0;j<b;j++){
    m = p + j*((n-p)/b);
    k = m;
    for(i=m;i<m+((n-p)/b);i++)
       if(d[nd[i]] < d[nd[k]]) k=i;
    ind[j] = k;
    val[j] = d[k];
 }
 k = 0;
 for(j=0;j<b;i++)
   if(val[j] < d[k]) k=j; 
 k = ind[k]
\end{Verbatim}
\caption{\label{psel}Parallelizable version of the selection process}
\end{figure}
We think that this transformation can be generalized in the context of greedy algorithms as the selection process is typically based on an \underline{associative} operation. The loop transformation described above is essentially based on the associativity of the underlying operation. We finally the OpenMP code provided in figure \ref{gsp-omp}.
 \begin{figure}[H]
\begin{Verbatim}[fontsize=\small]
for(p=1;p<n;p++){
 #pragma omp parallel for  private(s,k,i)
 for(j=0;j<b;j++){
    s = p + j*((n-p)/b);
    k = s;
    for(i=s;i<s+((n-p)/b);i++)
      if(d[nd[i]] < d[nd[k]]) k=i;
    ind[j] = k;
    val[j] = d[k];
 }
 k = 0;
 for(j=0;j<b;j++)
    if(val[j] < d[k]) k=j; 
 k = ind[k]
 swap(nd[k],nd[p]);
 k = nd[p]
 #pragma omp parallel for
 for(j=0;j<deg[k];j++){
    i = g[k][j];
    if(d[i] > d[k]+m[i,k]) d[i] = d[k]+m[i,k];
 }
 }
\end{Verbatim}
\caption{\label{gsp-omp}Dijkstra greedy algorithm for the shortest paths}
\end{figure}
\subsection{Minimum Spanning Tree}
A spanning tree of given a weighted undirected graph $G$ is a subgraph $H$ such that
\begin{itemize}
\item $H$ is a subgraph of $G$ (every edge of $H$ belongs to $G$)
\item $H$ spans $G$ (they have the same set of vertices)
\item $H$ is a tree (connected and acyclic)
\end{itemize}
A {\em minimum spanning tree} (MST) is a spanning tree with the minimum cost (sum of the weights of all the edges) among all possible spanning trees. Prim greedy algorithm to build an MST starts with a single node and iteratively select an external node with the minimum distance to the current MST. The algorithm can be written as displayed in figure \ref{gprim}. Node 0 is assume to be the first one to be selected, hence its distance is initially set to 0 and the others to infinity (array {\tt d}). As with the shortest paths, as we want to have the range {\tt [0..p-1]} (resp. {\tt [p..n-1]}) for the set of selected (resp. remaining) nodes, we use an array {\tt nd} such that {\tt nd[k]} is the id of the node at position {\tt k}.
 \begin{figure}[H]
\begin{Verbatim}[fontsize=\small]
d[0] = 0
for(i=1;i<n;i++) d[i] = INFINITY
for(p=0;p<n;p++){
    k = p
    for(i=p;i<n;i++)
       if(d[nd[i]] < d[nd[k]]) k = i;
    swap(nd[k],nd[p]);
    k = nd[p]
    for(j=0;j<deg[k];j++){
       i = g[k][j];
       if(d[i] > m[k,i]) d[i] = m[i,k];
    }
 }
\end{Verbatim}
\caption{\label{gprim}Prim greedy algorithm for the MST}
\end{figure}
We realize that Prim and Dijkstra algorithms have exactly the same structure, thus the same parallelization remarks and techniques applied. This holds also for Moore-Dijkstra algorithm for dominated graph flooding.
\subsection{Technical observations about the update procedure}
Note that in all of our selected greedy algorithms, the update procedure is relevant only for the neighborhood of the currently selected node. Thus, we can use:
\begin{itemize}
\item a binary array {\tt sel} to indicate whether a node has yet been selected or not
\item an array {\tt deg} for the degrees of the vertices
\item a array {\tt g[]} for the neighborhood of the vertices ({\tt g[i][j]}  is the $j{th}$ neighbor of node $i$)
\item a array {\tt m[]} for the weighted of the neighborhood of the vertices ({\tt m[i][j]}  is the weight of the edge with the $j{th}$ neighbor of node $i$)
\end{itemize}
The update procedure is now as in figure \ref{opup}.
 \begin{figure}[H]
\begin{Verbatim}[fontsize=\small]
for(j=0;j<deg[k];j++){
   i = g[k][j];
   if((sel[i]==0)&&(d[i]>m[k][j])) d[i] = m[k][j];
 }
\end{Verbatim}
\caption{\label{opup}Optimal form of the update process}
\end{figure}
The loop is still fully parallel, but its parallelization should be managed dynamically. Indeed, the parallelization is worth considering only if {\tt deg[k]} is large enough. The solution here is to estimate a convenient number of threads to be used for the loop parallelization directive. We can consider the sequence in figure \ref{par_opup}, where {\tt  c} is an arbitrary chunk size (to be evaluated empirically).
 \begin{figure}[H]
\begin{Verbatim}[fontsize=\small]
np = min(ceil(deg[k]∕c), nthreads
#pragma omp parallel for num_threads(np)
for(j=0;j<deg[k];j++){
   i = g[k][j];
   if((sel[i]==0)&&(d[i]>m[k][j])) 
                      d[i] = m[k][j];
 }
\end{Verbatim}
\caption{\label{par_opup}Optimal parallel update procedure}
\end{figure}
\subsection{Performance evaluation and related observations}
We evaluate our methodology for the greedy algorithms paradigm on the MST case using height 2.20 GHz-cores of an INTEL Broadwell E/P. The graph is randomly generated with various sizes and different levels of density. Table \ref{res_mst} displays our experimental results 
\begin{table}[H]
\begin{minipage}[l]{1.0\columnwidth}
\begin{tabular}{|r|r|r||r|r|r|r|r|r|r|r|r|}
\cline{5-11}
\multicolumn{4}{c}{} & \multicolumn{7}{|c|}{Number of cores (speedup)}\\   
\hline
{\bf N°} & {\bf N} & {\bf Degrees} & {\bf Seq T(s)} & {\bf 2} & {\bf 3} & {\bf 4} & {\bf 5} & {\bf 6} & {\bf 7}& {\bf 8}\\
\hline
1 & $10^5$ & [20:100] & 4.152 & 1.90 & 2.72 & 3.46 & 4.13 & 4.53 & 5.00 & 5.46 \\
\hline
2 & $10^5$ & [10:20] & 4.107 & 1.93 & 2.79 &  3.56 & 4.24 & 4.79 & 5.29 & 5.53 \\
\hline
3 & $2\times 10^5$ & [10:20] & 16.283 & 1.96 & 2.88 & 3.77 & 4.58 & 5.35 & 6.02 & 6.63 \\
4 & $4\times 10^5$ & [10:20] & 64.689 & 1.97 & 2.93 & 3.85 & 4.74 & 5.59 & 6.39 & 7.21 \\
\hline
\end{tabular}
\caption{\label{res_mst}Experimental results of our MST parallelization}
\end{minipage}
\end{table}

We can see that speedups are quite good and do no really depend on the density of the graph. Indeed, the processing time for the update procedure is negligible compare to that of the selection. We even realize that using several threads for this step is likely to degrade the speedup, certainly because we only paid for threads creation and management. However, we still recommend to keep in mind our initial analysis as some other problems might raise different complexity profiles.  
\section{Conclusion}\label{conc}
Our aim in this work was to study the parallelization of dynamic programming and greedy algorithms using directives-based paradigm. The motivation is that most of shared memory parallelizations are made through OpenMP and our two algorithmic paradigms cover a wide range of important combinatorial optimization kernels. It looks clear that applying some loop transformations is necessary in order to create or improve the parallelization potential of the original code. Depending on the specificities of the considered paradigm and the input scenario, it might useful to control the number of working threads in order to avoid speedup degradation due to the overhead of the parallelization. This discussion is important, especially in the perspective of manycore implementation with a larger number of threads. This is exactly what we are now investigating next to the current work, taking into account the challenging aspect of NUMA configurations.

\end{document}